\documentclass[submission,copyright,creativecommons]{eptcs}
\providecommand{\event}{QAPL 2011} 
\usepackage[small]{caption}

\newcommand{\captionfonts}{\footnotesize }

\makeatletter  
\long\def\@makecaption#1#2{%
  \vskip\abovecaptionskip
  \sbox\@tempboxa{{\captionfonts #1: #2}}%
  \ifdim \wd\@tempboxa >\hsize
    {\captionfonts #1: #2\par}
  \else
    \hbox to\hsize{\hfil\box\@tempboxa\hfil}%
  \fi
  \vskip\belowcaptionskip}
\makeatother   

\usepackage{amsmath,amssymb,amsfonts} 
\usepackage{epsfig,graphicx}

\usepackage[latin1]{inputenc}
\usepackage[english]{babel}
\usepackage[all]{xy}
\usepackage{colortbl}
\def\jo#1{\textcolor{black}{#1}}
\def\LEFT{{\mathrm{left}}}
\def\RIGHT{{\mathrm{right}}}

\newtheorem{theorem}{Theorem}

\newtheorem{proposition}{Proposition}
\newtheorem{definition}{Definition}
\newtheorem{remark}{Remark}
\newtheorem{example}{Example}

\newenvironment{proof}{{\bf ~Proof}.\enspace}%
{$\Box$}

\def\ang#1#2{\langle#1\rangle_{#2}\, }

\def\angs#1{\langle#1\rangle\, }
\def\cS{\mbox{Sub}}

\def\cL{\mathcal{L}}
\def\cF{\mathcal{F}}

\def\sem#1{[\hspace{-.35ex}[#1]\hspace{-.35ex}]}

\begin{document}  

\title{Computing Distances between Probabilistic Automata}
\def\titlerunning{Computing Distances between Probabilistic Automata}
\def\authorrunning{Tracol, Desharnais, Zhioua}

\author{
Mathieu Tracol \institute{LRI, Universit\'e Paris-Sud, France}
\and
Jos\'ee Desharnais\thanks{Research supported by CRSNG},
\quad 
Abir Zhioua$^*$
\institute{D\'ep.\ d'informatique et de g\'enie logiciel, Universit\'e Laval, Qu\'ebec, Canada}}
\maketitle

\begin{abstract}
We present relaxed notions of simulation and bisimulation on  Probabilistic Automata (PA), that allow some error $\epsilon$. When $\epsilon=0$ we retrieve the usual notions of bisimulation and simulation on PAs.
We give logical characterisations of these notions by choosing suitable  logics which differ from the elementary ones, $\cL$ and $\cL^\neg$, by the modal operator. 
Using flow networks, we show how to compute the relations in PTIME. This allows the definition of an efficiently computable
non-discounted distance between the states of a PA.  A natural modification of this distance is introduced, to obtain a discounted distance,  which weakens the influence of long
term transitions.  We compare  our notions of distance to others
previously defined and illustrate our approach on various examples.
We also show that our distance is not expansive with respect to process algebra operators.

 Although $\cL^{(\neg)}$ is a suitable logic to characterise $\epsilon$-(bi)simulation on deterministic PAs, it is not for general PAs;     interestingly, we prove that it does characterise  weaker notions, called \emph{a priori $\epsilon$-(bi)simulation}, which we prove to be NP-difficult to decide. 

\end{abstract}


\section{Introduction}

Preorders and equivalence notions between processes are central to concurrency theory.  One wants to compare terms of a process algebra for proving an axiomatisation sound, to compare processes to some  abstractions of them, etc. For non-probabilistic  processes, notions of bisimulation and simulation are widely acknowledged, with, of course, many variations.
In the study of probabilistic systems
it has been observed~\cite{Giacalone90}  that
 the comparison between processes should not be based on notions that rely strongly on exact numbers, as do  the known notions  of bisimulation and simulation for probabilistic systems.
The most important reason is that the stochastic information in probabilistic processes often comes from \emph{observations}, 
or from \emph{theoretical estimations}.  
  Hence a slight difference
in the probabilities between two processes should be treated differently from important ones and certainly not be simply tagged as non equivalence. 
In this context, notions of approximate equivalence or distance  are more useful.  Distances have been defined for probabilistic processes~\cite{Ferns05,Worrell01} and some have tried to estimate bisimulation with a certain degree of confidence~\cite{Ferns05}.  
  Relaxing the definition of simulation and bisimulation is another avenue, which we follow.
 
 We first extend previous work on deterministic processes~\cite{DesLavTra08}  to their non deterministic version, Probabilistic Automata (PA)~\cite{Segala94}.  We present \emph{relaxed} notions of simulation and bisimulation on  them with respect to some accuracy $\epsilon$. 
When $\epsilon=0$ we retrieve the usual notions of bisimulation and simulation on PAs. Our notions rely on a definition of \emph{$\epsilon$-lifting} of relations, which happens to be equivalent to the one presented in \cite{SegalaT07}. However, in this paper, the authors present different notions of $\epsilon$-simulations which consider \emph{distributions on the set executions}, whereas our relations are always between the \emph{states} of the systems, and our purpose is different. 
We give logical characterisations of these notions:  a state $\epsilon$-simulates another state if and only if it \emph{$\epsilon$-satisfies}  every formula that the other one (exactly) satisfies; similarly for $\epsilon$-bisimulation.  The extension of previous work comprises also the definition of an  efficiently computable
\emph{non-discounted} distance: two states are at distance less than or equal to $\epsilon$ if they are $\epsilon$-bisimilar. 
Using flow networks, we show how to compute in PTIME our relaxed relations
of (bi)simulation which helps to also compute efficiently the distance.

The nature of non determinism leads to new challenges and concepts.  It is not suprising that the logics that we prove to characterise $\epsilon$-bisimulation and $\epsilon$-simulation differ from the elementary ones, $\cL$ and $\cL^\neg$.   Although $\cL^{(\neg)}$ is a suitable logic to characterise $\epsilon$-(bi)simulation on deterministic PAs, it is not for general PAs;     interestingly, we define weaker notions  that it does characterise on PAs, called \emph{a priori $\epsilon$-(bi)simulation}. We also prove that a priori $0$-simulation  is NP-difficult to decide, contrarily to $\epsilon$-bi/simulation. 

We propose a natural modification of our basic distance in order to discount the influence of long
term transitions.  We  illustrate the difference between the values of the two distances on various examples of
two-dimensional grids. 
  Both (pseudo-)distances  are different  from the ones defined in the 
past~\cite{Desharnais99b,Alfaro07,Ferns05,Worrell01}, in that differences along paths are not accumulated, even in the discounted one.   The other known distances all accumulate differences through paths, and most of them discount the future.  Those that do not discount the future are intractable: it has recently been proven decidable~\cite{breugelSW07}, but with double exponential complexity.  Our distance is determined with a polynomial algorithm.

Finally, we prove that our distances are not expansive   with respect to process algebras operators, such as parallel  composition and non-deterministic choice.



\section{Probabilistic Automata and $\epsilon$-relations}

In this section we give the definitions of our models and the relaxed relations that we study.  
\emph{Probabilistic Automata} are labelled transition systems where transitions are from states to distributions and that involve non determinism. We generalize slightly the standard model, allowing sub-distributions instead of distributions, to model non responsiveness of the system and to make simulation a richer notion.
 Given a countable set $S$, we write $\cS(S)$ for the set of sub-distributions on $S$: the total probability out of a sub-distribution may be less than one. 
 Given a relation $R$ on $S\times S$ and $X\subseteq S$,  $R(X)=\lbrace y\in S|\exists x\in X\ s.t.\ xRy\rbrace$.
	 A set $X$ is \emph{$R$-closed} if $R(X)\subseteq X$.
\begin{definition}[PA~\cite{Segala94}]
A \emph{probabilistic automaton}, or PA, is a tuple $\mathcal S=(S,Act,\mathcal{D})$ where $S$ is a denumerable state space, $Act$ is a finite set of actions, and $\mathcal{D}\subseteq S\times Act\times \cS(S)$ is the transition relation. $\mathcal S$ is \emph{finitely branching} if for all $s\in S$ and $a\in Act$, $\{\mu\in\cS(S)\mid (s,a,\mu)\in \mathcal D)\}$ is finite; if it is a singleton or empty, we say that $\mathcal S$ is \emph{deterministic}. The disjoint union  of PAs $\mathcal{S}_1,...,\mathcal{S}_k$ is the  PA $\uplus_{i\in[1;k]} \mathcal S_i$ whose states are the disjoint union of the ${S}_i$ and transitions carry through.

\end{definition}
 Closely related models restrict states to be either probabilistic or non deterministic~\cite{HermannsIMC02}.  Generalizations to uncountable state spaces have also been studied~\cite{CSKN05,DArg09}.  
We sometimes mark a state (or a distribution) as initial. We write $s\stackrel{a}{\rightarrow}\mu$ for a transition $(s,a,\mu)\in\mathcal{D}$. 

An example of PA is given in Fig.~\ref{fig:basicPA}:  an arrow labelled with  action $l$ and value $r$  represents an $l$-transition of probability $r$; in picture representations, 
we omit the distributions that are concentrated in one point, as can be seen for the transitions from $s$ to $s_0$ and from $t$ to itself.  In contrast, state $s$ has an $a$-transition to distribution $\mu_s$ giving $s$ three possible successors for this $a$-transition.
\begin{figure}[t]
$$\xymatrix@R=3mm@C=10mm{
&& s\ar@{-}[d]^{a}\ar[llddd]_{b,\frac{1}{2}}\ar@(dr,ur)[]_{a,1}
&& t \ar@{-}[d]_{a}\ar[rrddd]^{b,\frac{5}{8}}\ar@(dl,ul)[]^{a,1}\\
&& \scriptstyle\mu_s\ar[ddl]_{\frac{1}{8}} \ar[dd]^{\frac{3}{8}} \ar[ddr]^{\frac{4}{8}}  
&& \scriptstyle\mu_t\ar[ddl]_{\frac{2}{8}}\ar[dd]_{\frac{4}{8}}\ar[ddr]^{\frac{2}{8}}\\\\
s_0&     s_1&s_2 \ar[dd]_{b,1}  & x\ar[dd]_{b,1}  &  t_2\ar[dd]^{b,1}&t_1 &t_0\\\\
&     & s_3\ar[r]^{c,1}   & y              &  t_3\ar[l]_{c,1}
  }
$$
\caption{$s\prec_{\frac{1}{8}}t$ and  $t\prec_{\frac{1}{8}}s$ but $s\not\sim_{\frac{1}{8}}t$}
\label{fig:basicPA}
\end{figure}
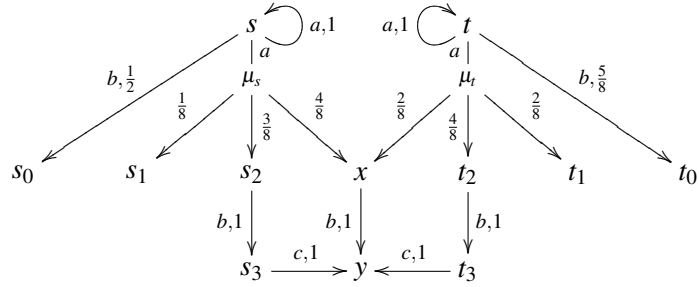


 In previous work~\cite{DesLavTra08}, we relaxed the classical notion of simulation between deterministic PAs to \emph{$\epsilon$-simulation}. We now generalize this approach to the context of PAs. 



\begin{definition}[$\epsilon$-bi/simulation]\label{d:eSim}
Let  $\mathcal{S}=(S,Act,\mathcal{D})$ be a PA, and $\epsilon\geq0$. A relation $R\subseteq S\times S$ is an \emph{$\epsilon$-simulation on $\mathcal{S}$} if whenever $sRt$, if $s\stackrel{a}{\rightarrow}\mu$, then there exists a transition $t\stackrel{a}{\rightarrow}\nu$ such that $\mu\,\mathcal{L}^\epsilon(R)\,\nu$, where 
$$\mu\,\mathcal{L}^\epsilon(R)\,\nu\quad\mbox{ iff \,\,for all }E\subseteq S\mbox{ we have }\mu(E)\leq\nu(R(E))+\epsilon.$$
 If $R$ is  symmetric, it is an  \emph{$\epsilon$-bisimulation}. 
{Two states $s$ and $s'$ of PAs $\mathcal S_1$ and $\mathcal S_2$ are $\epsilon$-similar (resp.~$\epsilon$-bisimilar), written $s\prec_\epsilon s'$  (resp.~ $s\sim_\epsilon s'$), if there is some  $\epsilon$-simulation  (resp.~$\epsilon$-bisimulation) that relates them in $\mathcal S_1\uplus \mathcal S_2$. }

\end{definition}

We may omit $\epsilon$ in the notation when $\epsilon=0$, as it yields the classical notions.  

\begin{example}\label{e:twoway}
  In the PA of Fig.~\ref{fig:basicPA} $s\prec_{\frac{1}{8}}t$ and  $t\prec_{\frac{1}{8}}s$.  This is witnessed by
the relations 
$$\{(s,t),(x,x), (x,t_2),(y,y), (y,t_3)\}\cup \{(s_i,t_i)\mid i=0,1,2,3\},$$
$$\{(t,s),(x,x), (x,s_2),(y,y),(y,s_3)\}\cup \{(t_i,s_i)\mid i=0,1,2,3\}.$$
    However, $s$ and $t$ are not related for any $\epsilon<\frac{1}{8}$.
Notice that we have $s\not\sim_{\frac{1}{8}}t$. Indeed, $x$ is ${\frac{1}{8}}$-bisimilar to no state but itself, but $\mu_s(\left\{x\right\})=\frac{4}{8}$, which is strictly greater than $\mu_t(\left\{x\right\})+\frac{1}{8}=\frac{3}{8}$.  
\end{example}
This example shows that two-way $\epsilon$-simulation is not $\epsilon$-bisimulation, even for deterministic PAs. A deterministic  example is obtained by removing $a$-loops in Fig.~\ref{fig:basicPA}.

\begin{proposition}
$\epsilon$-bisimulation is different from two-way $\epsilon$-simulation.\end{proposition}




As for the classical case, we define the largest relations as greatest fixed~points: $F^\epsilon:2^{S\times S}\to 2^{S\times S}$ is defined as follows $\forall R\subseteq S\times S,\ \forall s,t\in S$, let $(s,t)\in F^\epsilon(R)$ iff
$\forall s\stackrel{a}{\rightarrow}\mu,\ \exists t\stackrel{a}{\rightarrow}\nu\ |\ \mu\mathcal{L}^\epsilon(R)\nu$.
Similarly, $G^\epsilon:2^{S\times S}\to 2^{S\times S}$ is defined as $(s,t)\in G^\epsilon(R)$ iff $(s,t),(t,s)\in F^\epsilon(R)$. We then define $\prec_{\epsilon}^n$ iteratively as $\prec_{\epsilon}^0=S\times S$ and for all $n$, $\prec_{\epsilon}^{n+1}=F^\epsilon(\prec^n)$. As well, let $\sim_{\epsilon}^0=S\times S$ and for all $n$, $\sim_{\epsilon}^{n+1}=G^\epsilon(\sim^n)$.
\begin{theorem} $\mathcal{S}$ being finitely branching,
$\prec_\epsilon$ and $\sim_\epsilon$ are the greatest fixpoints of $F^\epsilon$ and $G^\epsilon$, respectively. In other words, $\prec_{\epsilon}=\bigcap_{n\in\mathbb{N}}\prec_{\epsilon}^n,\ \ \mbox{and}\ \ \sim_{\epsilon}=\bigcap_{n\in\mathbb{N}}\sim_{\epsilon}^n$.
\end{theorem}

\subsection{Lifting of relations and flow networks.} 

The lifting of a relation is a standard construction that transfers a relation $R$ on states to a relation $\mathcal{L}(R)$ on sub-distributions over states. Contrarily to the way we formulate Def.~\ref{d:eSim},  the usual definition of liftings is rather in terms of the existence of a weight function~\cite{Segala94}. 
We show the equivalence of the definitions.

\begin{definition}[$\epsilon$-weight functions]
Let $\epsilon\geq0$, and $\mu,\nu\in\cS(S)$. An \emph{$\epsilon$-weight function for $(\mu,\nu)$ with respect to $R$} is a function $\delta: S\times S\rightarrow[0;1]$ such that:
\begin{itemize}
 \item If $\delta(s,t)>0$ then $(s,t)\in R$.
 \item For all $s,t\in S$, $\sum_{s'\in S}\delta(s,s')\leq\mu(s)$ and $\sum_{s'\in S}\delta(s',t)\leq\nu(t)$.
 \item $\sum_{s,s'\in S}\delta(s,s')\geq\mu(S)-\epsilon$.
\end{itemize}
\end{definition}


Before stating the equivalence between our formulation of $\mathcal{L}^\epsilon(R)$ and the one with weight functions, we recall the notion of flow network,  since it provides a convenient alternative definition for applications~\cite{Baier96b}.  

A network is a tuple $\mathcal{N}=(V,E,\perp,\top,c)$ where $(V,E)$ is a finite directed graph in which every edge $(u,v)\in E$ has a non-negative, real-valued capacity $c(u,v)$. If $(u,v)\not\in E$ we assume $c(u,v)=0$. We distinguish two vertices: a source $\perp$ and a sink $\top$. For $v\in V$ let $in(v)$ be the set of incoming edges to node $v$, and $out(v)$ the set of outgoing edges from node $v$. A \emph{flow function} is a real function $f:V \times V \rightarrow \mathbb{R}$ with the two following  properties for all nodes $u$ and $v$:
\begin{itemize}
	\item Capacity constraints: $0\leq f(u,v)\leq c(u,v)$. The flow along an edge cannot exceed its capacity.
	\item Flow conservation: for each node $v\in V-\left\{\perp,\top\right\}$, we have $\sum_{e \in in(v)} f(e) = \sum_{e \in out(v)} f(e)$.
\end{itemize}
The flow $\mathcal{F}(f)$ of $f$ is given by $\mathcal{F}(f)=\sum_{e \in out(\perp)} f(e)-\sum_{e \in in(\top)} f(e)$. 

\begin{definition}[The network $\mathcal{N}(\mu,\nu,R)$]
Let $S$ be a finite set, $R\subseteq S\times S$, and $\mu,\nu\in \cS(S)$. Let $S'=\left\{t'|t\in S\right\}$, where $t'$ are pairwise distinct ``new" states (i.e. $t'\not\in S$). Let $\perp$ and $\top$ be two distinct new elements not contained in $S\cup S'$. The network $\mathcal{N}(\mu,\nu,R)=(V,E,\perp,\top,c)$ is defined as follows:
\begin{itemize}\addtolength{\itemsep}{-3pt}
	\item $V=S\cup S'\cup \left\{\perp,\top\right\}$.
	\item $E=\left\{(s,t')|(s,t)\in R\right\}\cup\left\{(\perp,s)|s\in S\right\}\cup\left\{(t',\top)|t\in S\right\}$.
	\item The capacity function $c$ is given by: $c(\perp,s)=\mu(s)$, $c(t,\top)=\nu(t)$, and $c(s,t)=1$ for all $s,t\in S$. 
\end{itemize}
\end{definition}

The following proposition gives various characterizations of the simulation relation.

\begin{proposition}\label{l-difflift}
Let $S$ be a finite set, $R\subseteq S\times S$, and $\mu,\nu\in \cS(S)$. The following properties are equivalent:
\begin{enumerate}\addtolength{\itemsep}{-3pt}
 \item $\mu\mathcal{L}^\epsilon(R)\nu$.
 \item The maximal flow in $\mathcal{N}(\mu,\nu,R)$ is greater than or equal to $\mu(S)-\epsilon$.
 \item There exists an $\epsilon$-weight function for $(\mu,\nu)$ with respect to $R$.
\item For all $R$-closed set $E\subseteq S$, we have $\mu(E)\leq\nu(E)+\epsilon$.
 \end{enumerate}
The equivalence with 4 applies only if the domain and image of $R$ are considered disjoint.  

\end{proposition}
\begin{proof}
$1\Leftrightarrow2$ is Theorem 7 \cite{DesLavTra08}.  $2\Leftrightarrow3$ is  a slight generalization of a result of \cite{Baier96b}, in which $\epsilon=0$.
$1\Rightarrow4$ is always true and   $4\Rightarrow1$ is straightforward if the domain and image of $R$ are disjoint.
\end{proof}

The condition on the fourth statement may look restrictive, but it is quite natural. By taking two copies of the state space and relating a state to  the copies of the states that $R$ relates it to, one obtains a relation that satisfies the condition, yet representing the same relation as $R$.  This allows to make a distinction between states that are \emph{simulated} from states that are viewed as \emph{simulating}.

\section{Logic for $\epsilon$-simulation and $\epsilon$-bisimulation}

\subsection{The logic $\mathcal{L}$ and its corresponding notion of simulation}
In the context of deterministic PAs,  $\epsilon$-bi/simulation  are characterized~\cite{DesLavTra08} by the simple logic $\mathcal{L}^\neg$~\cite{Larsen91}, using a relaxed semantics  ``up to $\epsilon$".

\begin{definition}[The logics $\mathcal{L}$ and $\mathcal{L}^\neg$]
 The syntax of $\mathcal{L}^\neg$ is as follows:
\begin{center}
 $\mathcal{L}^\neg\quad \phi::=\top\mid \neg\phi |\ \phi_1\wedge\phi_2  |\ \phi_1\vee\phi_2  \ |\ \ang a\delta\phi$ ~~where $\delta\in\mathbb{Q}\cap[0;1]$.
\end{center}

We write $\mathcal{L}$ for the logic  without negation. Given a PA with components $(S,Act,\mathcal{D})$, the relaxed $\epsilon$-semantics $\models_\epsilon$ is defined by structural induction on the formulas. 
\begin{itemize}
 \item \hspace{-3mm} $\forall s\in S,s\models_\epsilon\top \quad\bullet s\models_\epsilon\neg\theta\;$ iff $\;s\not\models_{-\epsilon}\theta$ $\quad\bullet$ $s\models_\epsilon\phi_1\wedge\phi_2\;$ iff $\;s\models_\epsilon\phi_1$ and $s\models_\epsilon\phi_2$;
\item\hspace{-2mm} $s\models_\epsilon \ang a\delta\phi\ \mbox{iff  {there exists} a transition } s\stackrel{a}{\rightarrow}\mu\ \mbox{s.t. } \mu(\sem{\phi}_\epsilon)\geq\delta-\epsilon$,
\end{itemize}
where  ${\sem \phi}_\epsilon =\lbrace s\in S|s\models_\epsilon\phi\rbrace$, and the semantics of $\vee$ is similar to the one of $\wedge$.  Given $s\in S$, we write $\mathcal{F}^\epsilon(s)$ (resp. $\mathcal{F}^{\epsilon,\neg}(s)$) for the set of formulas in $\mathcal{L}$ (resp. in $\mathcal{L}^\neg$) $\epsilon$-satisfied by $s$.
\end{definition}

\jo{As for deterministic PAs,  the logic is less expressive with the relaxed semantics~\cite{DesLavTra08} than with the standard one. Indeed, for each $\phi\in\mathcal{L}^{\neg}$, we can construct an associated formula $\phi_\epsilon\in\mathcal{L}^\neg$ such that $\sem\phi_\epsilon=\sem{\phi_\epsilon}_0$.  Here is how this is done: $\top_\epsilon=\top$; $(\phi_1\wedge\phi_2)_\epsilon=(\phi_1)_\epsilon\wedge(\phi_2)_\epsilon$; $(\ang a\delta\phi)_\epsilon=\ang a{\delta-\epsilon}\phi_\epsilon$; $(\neg\phi)_\epsilon=\neg(\phi_{-\epsilon})$.
Here we use the fact that $\ang a\lambda\phi$ is still a valid formula, even if $\lambda<0$ or $\lambda>1$, which gives in turn that $(\phi_\epsilon)_{-\epsilon}=\phi$.
Clearly the transformation is additive, as  $(\phi_\epsilon)_{\epsilon^{'}}=\phi_{\epsilon+\epsilon^{'}}$.
}

\begin{example}
 If $\phi=\ang a{.5}(\neg\ang a{.2}\top)$, then $\phi_\epsilon=\ang a{.5-\epsilon}(\neg\ang a{.2+\epsilon}\top)$.
\end{example}

 The relaxed logic being less expressive is not an issue because  we  use the new semantics to simplify the formulations of the  logical characterisations.  It implies that model checking of formulas with the relaxed semantics can be done using the same technique as for the usual semantics.

The logics $\mathcal{L}$ and $\mathcal{L}^{\neg}$ induce $\epsilon$-simulation and $\epsilon$-bisimulation relations: 

\begin{definition}[Logical $\epsilon$-simulation.]\label{d:logicalSim}
Let $\epsilon\in[0;1]$,  $s,t\in S$. We say that \emph{$t$ $\mathcal{L}_\epsilon$-simulates $s$}, written $s\prec_{\epsilon}^\mathcal{L} t$, if for all formula $\phi\in\mathcal{L}$, $s\models \phi$ implies $t\models_\epsilon\phi$. States $t$ and $s$ are said \emph{$\epsilon$-logically equivalent,} written $s\sim^{\mathcal{L}^\neg}_\epsilon t$, if $s\prec_\epsilon^{\mathcal{L}^\neg} t$ and $t\prec_\epsilon^{\mathcal{L}^\neg} s$.


\end{definition}

The following theorem says that $\mathcal{L}$ characterizes $\epsilon$-simulation on deterministic PAs.  As Ex.~\ref{e:twoway} illustrates,  $\epsilon$-bisimulation is different from two-way $\epsilon$-simulation when $\epsilon>0$, and hence we need  negation to characterize $\epsilon$-bisimulation. 

\begin{theorem}[\cite{DesLavTra08}]
For deterministic PAs we have 
 $\prec_\epsilon\,=\,\prec^{\mathcal{L}}_\epsilon$  and
 $\sim_\epsilon\,=\,\sim^{\mathcal{L}^\neg}_\epsilon$
\end{theorem}

As can be expected, the logics $\mathcal{L}$ and $\mathcal{L}^{\neg}$  are not strong enough to characterize $\epsilon$-bi/simulation for PAs.  In the next section, we will present a stronger logic that will characterize these notions.  Nevertheless, we can  present the notions that correspond to $\prec^{\mathcal{L}}_\epsilon$ and $\sim^{\mathcal{L}^\neg}_\epsilon$.  The name  ``a-priori"~\cite{Alfaro07} comes from the order of the quantifiers in the definition: sets $E$ in the following definition are chosen \emph{before} the matching transition, which contrasts with  Def.~\ref{d:eSim}.

\begin{definition}[A priori $\epsilon$-simulation and bisimulation]
Let $\mathcal{S}=(S,Act,\mathcal{D})$ be a PA. A relation $R\subseteq S\times S$ is an \emph{a priori $\epsilon$-simulation} on $\mathcal{S}$ iff for all $s,t\in S$ such that $sRt$, for all $s\stackrel{a}{\rightarrow}\mu$ and for all $E\subseteq S$, there exists a transition $t\stackrel{a}{\rightarrow}\nu$ such that $\mu(E)\leq\nu(R(E))+\epsilon$. 
If $R$ is symmetric, then it is an  \emph{a priori $\epsilon$-bisimulation}. We write $\preceq_\epsilon^{prio}$ and $\sim_\epsilon^{prio}$ for the largest relations of a priori $\epsilon$-simula\-tion and $\epsilon$-bisimula\-tion.
\end{definition}

Before proving that this relation is characterized by $\cL$, we introduce some notation.
We define $\prec^{prio,n}_{\epsilon}$ iteratively in the same way as $\prec_{\epsilon}^n$, using $F^{prio}_\epsilon:2^{S\times S}\to 2^{S\times S}$, which is defined as follows: $\forall R\subseteq S\times S,\ \forall s,t\in S$, let $(s,t)\in F^{prio}_\epsilon(R)$ iff
$\forall s\stackrel{a}{\rightarrow}\mu,\ \forall E\subseteq S,\ \exists t\stackrel{a}{\rightarrow}\nu\ s.t.\ \mu(E)\leq\nu(R(E))+\epsilon$.
As for $\prec_{\epsilon}=\bigcap_{n\in\mathbb{N}}\prec^{n}_{\epsilon}$, we can show that we have $\prec^{prio}_{\epsilon}=\bigcap_{n\in\mathbb{N}}\prec^{prio,n}_{\epsilon}$.
	The \emph{depth} of a formula is the maximal number of imbrications of $\ang a \delta$ operators. We write $\mathcal{F}^n$ for the set of formulas of $\mathcal{L}$ of depth at most $n$.
	 Given $n\in\mathbb{N}$,  given $s\in S$, $\mathcal{F}^n_\epsilon(s)$ is the set of formulas $\phi\in\mathcal{L}$ of depth at most $n$ such that $s\models_\epsilon\phi$ and $\mathcal{F}_\epsilon(s)$ is the set of formulas $\phi\in\mathcal{L}$ such that $s\models_\epsilon\phi$.  We define $\mathcal{F}^n(s)=\cF^n_0(s)$ and $\mathcal{F}(s)=\cF_0(s)$.

	
The next theorem proves the logical characterization of 	$\prec_\epsilon^{prio}$ and $\sim_\epsilon^{prio}$.
\begin{theorem}\label{th:priovspost}
Let $\mathcal{S}=(S,Act,\mathcal{D})$ be a PA, and let $s,t\in S$. Then:
\begin{enumerate}
\item $s\prec^{prio,n}_\epsilon t$ iff $\mathcal{F}^n(s)\subseteq\mathcal{F}_\epsilon^n(t)$ for all $n\geq0$
\item 	For all $n \geq 0$, for all $u,v\in S$, there exists $\phi_u\in\mathcal{F}_n$ such that 
	$v\models_\epsilon\phi_u\ \mbox{iff}\ u\prec^{prio,n}_\epsilon v$
\item $\prec_\epsilon\;\subseteq\;\prec_\epsilon^{prio} \;=\;\prec^{\mathcal{L}}_\epsilon$ and the inclusion is strict.
 \item $\sim_\epsilon\;\subseteq\;\sim_\epsilon^{prio} \;=\;\sim^{\mathcal{L}^\neg}_\epsilon $ and the inclusion is strict.
\end{enumerate}
\end{theorem}

This proof generalizes to the context of countable state space systems, using a method close to the method used in \cite{DesLavTra08} to extend logic characterizations to denumerable state spaces. 


\begin{proof}
Inclusions are straightforward. The proof of the strictness of inclusion is given by the example following the proof.  The structure of the proof is similar to the ones of \cite{Larsen91} and \cite{Segala94} for the logical characterisation of simulation but we adapt them to  systems with non determinism.   The third point is a corollary of the first one. The fourth point is not more difficult and is the same kind of translation as the proof of \cite{Larsen91}.  
We sketch the proof of the first two   points, concentrating on the ``$\Leftarrow$" direction.
The two points are proven simultaneously  by induction on $n$.  
The base case follows trivially from the definitions.
 Assume that the claims are true for $n$. \jo{We now prove 2 for $n+1$. Fix $u\in S$.  We define $\phi_u\in \mathcal F_{n+1}$ as follows.  Let $v\in S$ such that  $u\not\prec^{prio,n}_\epsilon v$; by induction $\mathcal{F}^n(u)\not\subseteq\mathcal{F}^n_\epsilon(v)$, that is, there exists a formula $\phi_{(u,v)}\in\mathcal{F}_n$, such that $u\models \phi_{(u,v)}$ and $v\not\models_\epsilon\phi_{(u,v)}$. The formula $\phi_u:=\bigwedge_{v\not\prec^{prio,n}_\epsilon u}\phi_{(u,v)}$ is in 
$\mathcal{F}_{n+1}$ because $S$ is finite. Now, $u\models\phi_u$, and any $v$ such that $u\not\prec^{prio,n+1}_\epsilon v$ verifies $u\not\prec^{prio,n}_\epsilon v$ and hence $v\not\models_\epsilon\phi_u$.}  Since $u\prec^{prio,n}_\epsilon w$ implies $w\models\phi_u$ by induction, we get the result; hence 2 is proven for $n+1$.  As for 1,  suppose $\mathcal{F}^{n+1}(s)\subseteq\mathcal{F}^{n+1}_\epsilon(t)$. Let $s\stackrel{a}{\rightarrow}\mu$ be a transition from $s$, and let $E\subseteq S$. We are looking for a transition $t\stackrel{a}{\rightarrow}\nu$ such that $\nu(\prec^{prio,n}_\epsilon E)\geq\mu(E)-\epsilon$. Let $p=\mu(E)$. We construct a formula $\phi_E$ such that for all $u\in S$, $u\models\phi_E$ iff $u\in E$. By the second claim, we just have to consider $\phi_E=\bigvee_{u\in E}\phi_u$. Now, $s\models\ang a p \phi_E$. Since $\ang a p \phi_E \in\mathcal{F}^{n+1}$, we must have $t\models_\epsilon\ang a p \phi_E$. Since $\prec^{prio,n}_\epsilon \sem {\phi_E} \subseteq {\sem \phi_E}_\epsilon$, we get that there exists a transition $t\stackrel{a}{\rightarrow}\nu$ from $t$ such that $\nu(\prec^{prio,n}_\epsilon E)\geq\mu(E)-\epsilon$, which proves the result.
\end{proof}

\begin{example}\label{e:prioNOTpost}
In the following PA (where the $b_i$'s are different labels), we have $s\prec^{prio} t$ and~$s\not\prec t$.     
$$\xymatrix@R=3mm@C=10mm{
&     s	\ar@{-}[d]^a																&&& t\ar@{-}[dl]_a\ar@{-}[d]^a\ar@{-}[dr]^a \\
	& \mu\ar[dl]_{\frac{1}{3}}\ar[d]^{\frac{1}{4}}\ar[dr]^{\frac{5}{12}} 	
		&&\nu_1\ar[drr]\ar[d]\ar[dr]&\nu_2\ar[dl]\ar[d]\ar[dr]&\nu_3\ar[dl]\ar[d]\ar[dll]\\
	s_1\ar[dd]_{b_1,1}&s_2\ar[dd]_{b_2,1}&s_3\ar[dd]_{b_3,1}               &  t_1\ar[dd]_{b_1,1}&t_2\ar[dd]_{b_2,1}&t_3\ar[dd]_{b_3,1}\\\\
	s_4&s_5&s_6								&      t_4&t_5&t_6
}~~~~~~ \begin{tabular}[t]{|c|c|c|c|}\hline
		&$t_1	$&$t_2	$&$t_3$\\\hline	\raisebox{-2.5mm}{\rule{0mm}{7mm}}
$\nu_1$&$\frac{7}{24}$&$\frac{7}{24}$&$\frac{5}{12}$\\\hline	\raisebox{-2.5mm}{\rule{0mm}{7mm}}
$\nu_2$&$\frac{3}{8}$&$\frac{1}{4}$&$\frac{3}{8}$\\\hline	\raisebox{-2.5mm}{\rule{0mm}{7mm}}
$\nu_3$&$\frac{1}{3}$&$\frac{1}{3}$&$\frac{1}{3}$\\\hline
\end{tabular}
$$
The transitions has been chosen such that  $\nu_3(t_1)=\mu(s_1)$,
 $\nu_2(t_2)=\mu(s_2)$,
 $\nu_1(t_3)=\mu(s_3)$,
 $\nu_1(\lbrace t_1\rbrace\cup\lbrace t_2\rbrace)=\mu(\lbrace s_1\rbrace\cup\lbrace s_2\rbrace)$,
 $\nu_2(\lbrace t_1\rbrace\cup\lbrace t_3\rbrace)=\mu(\lbrace s_1\rbrace\cup\lbrace s_3\rbrace)$,
 $\nu_3(\lbrace t_2\rbrace\cup\lbrace t_3\rbrace)=\mu(\lbrace s_2\rbrace\cup\lbrace s_3\rbrace)$.
Then it is easy to see that $s_i$ is simulated by $t_i$ for all $i=1,2,3$. Moreover, the last set of equalities shows that $s\prec_{prio} t$. However, we do not have $s\prec_{} t$. Indeed for all transitions $t\stackrel{a}{\rightarrow}\nu$ (combined or not) from $t$, we can find a set $E\subseteq \cup_{i\in\{1,2,3\}}\left\{s_i,t_i\right\}$ containing $s_i$ if and only if it contains $t_i$, and such that $\mu(E)>\nu(E)$.

\end{example}

\begin{remark}
By Theorem \ref{th:priovspost} item 3, for the PA of Ex.~\ref{e:twoway}, we have $s\prec_{\frac{1}{8}}^{prio}t$ and $t\prec_{\frac{1}{8}}^{prio}s$. Now, the only state $\frac{1}{8}$-a-priori bisimilar to $x$ is, here again, $x$ itself. Hence, for $s\stackrel{a}{\rightarrow}\mu_x$ and $E=\lbrace x\rbrace$, there exists no transition $t\stackrel{a}{\rightarrow}\nu$ such that $\mu(E)\leq\nu(\sim_{\frac{1}{8}}^{prio}(E))+\frac{1}{8}$, hence $s$ and $t$ are not $\frac{1}{8}$-a-priori bisimilar. As a consequence, two-way a-priori simulation is different from a-priori bisimulation, and the negation is needed in the logical characterization of bisimilarity.
\end{remark}



\textbf{Decidability of A Priori Simulation.} An interesting fact is that it is NP-hard to decide  a-priori simulation and bisimulation, even when $\epsilon=0$. This contrasts with classical results on strong simulation and bisimulation  whose decision procedures were proven to be in Poly-time (see \cite{Baier96b,zhang2007flow,Segala02}). The proof of the following theorem, not presented here due to lqck of space, is by reducing the \emph{subset sum problem}, known to be NP-complete (\cite{garey1979computers}), to our problem.

\begin{theorem}\label{t:aprioriNP}
The following problem  is NP-complete:\\
\textbf{Input:} A PA $\mathcal{S}$, $s,t\in S$.
\textbf{Question:} Do we have $s\prec^{prio} t$?
\end{theorem}

\subsection{The logic $\mathcal{L}^N$ for PAs}

We saw in the previous subsection that $\mathcal{L}$  is not strong enough for PAs.  We now give a logic characterizing our relaxed relations $\prec_\epsilon$ and $\sim_\epsilon$.
   The difference between this logic and $\mathcal{L}^{\neg}$ is the modal operator that permits to ``isolate" a distribution out of a state, and write properties that it satisfies.  This allows the semantics to be defined on states, as pointed out as well by D'Argenio et al.~\cite{DArg09}.  In contrast, Parma and Segala~\cite{Segala07} used a semantics on distributions to prove the logical characterisation of bisimulation (with $\epsilon=0$).

\begin{definition}[The logic $\mathcal{L}^{N,\neg}$]
 The syntax differs by one operator from  $\mathcal{L}^{\neg}$:
\begin{center}
 $\mathcal{L}^{N,\neg}~ \phi:=\top|~\neg\phi\mid \phi_1\wedge\phi_2  |\ \phi_1\vee\phi_2  \ |\ \angs a\lbrace(\phi_i,p_i)\rbrace_{i\in I}$ $\ \ \ $ $I$ finite, $p_i\in\mathbb{Q}\cap[0;1]$. \end{center}
We write $\mathcal{L}^N$ for the same logic without negation. 
Let $\epsilon\in[0;1]$.  
The relaxed semantics $\models_\epsilon$ of $\mathcal{L}^{N,\neg}$ is defined by structural induction on the formulas, in the same way as for $\cL^{\neg}$ except for the modal formula: $s\models_\epsilon \angs a\lbrace(\phi_i,p_i)\rbrace_{i\in I}$ iff there exists a transition $s\stackrel{a}{\rightarrow}\mu$ from $s$ such that for all $i\in I$, we have $\mu({\sem {\phi_i}}_\epsilon)\geq p_i-\epsilon$.
\end{definition}

As for the logic $\cL^{\neg}$  we can, by structural induction on the formulas, construct for each $\phi\in\mathcal{L}^{N,\neg}$ an associated formula $\phi_\epsilon\in\mathcal{L}^{N,\neg}$  such that ${\sem \phi}_\epsilon=\sem {\phi_\epsilon}$.  Hence, here again, model checking  of formulas with the relaxed semantics can be done using the same technique as for the usual semantics.
The logics $\mathcal{L}^N$ and $\mathcal{L}^{N,\neg}$ induce the  relations $\prec^{\mathcal{L}^N}_\epsilon $ and $\sim^{\mathcal{L}^N}_\epsilon $ as in Def.~\ref{d:logicalSim}. 
The following example illustrates how this logic differs from $\cL$. The key difference is that formula $\angs a\lbrace(\phi_i,p_i)\rbrace_{i\in I}$ is not equivalent to $\wedge_{i\in I}\ang a{p_i}{\phi_i}$.
\begin{example}
Consider the PAs of Ex.~\ref{e:prioNOTpost}.  As $s\prec^{prio} t$, we also have $s\prec^{\mathcal{L}}_\epsilon t$ by Theorem~\ref{th:priovspost}, and hence every formula of $\cL$ satisfied by $s$ is also satisfied by $t$.  However,  $s\not\prec t$ and the formula $$ \angs a\lbrace (\ang {b_1} 1 \!\top,\frac{1}{3}), (\ang {b_2} 1 \!\top,\frac{1}{4}),(\ang {b_3} 1 \!\top,\frac{5}{12}) \rbrace$$ of $\cL^N$ is not satisfied by $t$.  Note how the semantics forces the $\nu_i$'s to commit before the three ``$b_i$ formulas" are checked.
\end{example}

The following theorem is a logical characterization of the $\epsilon$-relations. Notice that  we need negation  once again, since two-way simulation is different from bisimulation.

\begin{theorem}\label{t-equiv sim ND}
 $\prec_{\epsilon}=\prec^{\mathcal{L}^N}_\epsilon $ and $\sim_{\epsilon}=\sim^{\mathcal{L}^{N,\neg}}_\epsilon$.
\end{theorem}
\begin{proof}
$[\prec_{\epsilon}\subseteq\prec^{\mathcal{L}^N}_\epsilon] $. 
Let $R$ be an $\epsilon$-simulation. We prove by structural induction that for all $\phi\in\mathcal{L}^N$, $R(\sem\phi)\subseteq {\sem\phi}_\epsilon$. We prove the case where $\phi=\angs a\lbrace(\phi_i,p_i)\rbrace_{i\in I}$, since the other cases are trivial. Let $s\in\sem\phi$, $t\in R(\{s\})$ and let $s\stackrel{a}{\rightarrow}\mu$ be the associated transition such that for all $i\in I$, $\mu({\sem {\phi_i}})\geq p_i$.
Since $R$ is an $\epsilon$-simulation, there exists a transition $t\stackrel{a}{\rightarrow}\nu$ such that $\mu\mathcal{L}^\epsilon(R)\nu$. Thus, for all $E\subseteq S$, we have $\mu(E)\leq\nu(R(E))+\epsilon$. In particular, given $i\in I$, we get that $\mu(\sem {\phi_i})\leq\nu(R(\sem {\phi_i}))+\epsilon$. But by induction hypothesis, we know that for all $i\in I$, $R(\sem{\phi_i})\subseteq{\sem{\phi_i}}_\epsilon$. This gives us that $\mu(\sem {\phi_i})\leq\nu({\sem {\phi_i}}_\epsilon))+\epsilon$ hence the result since by hypothesis $\mu(\sem {\phi_i})\geq p_i$.

$[\prec_{\epsilon}\supseteq\prec^{\mathcal{L}^N}_\epsilon] $.  We prove that $\prec_\epsilon^\mathcal{L}$ is an $\epsilon$-simulation.  Suppose that $s\prec_{\epsilon}^\mathcal{L} t$, and let $s\stackrel{a}{\rightarrow}\mu$ be a transition from $s$.
 We need to find some $t\to \nu$ such that $\nu(\prec_\epsilon^\mathcal{L}(X))\geq\mu(X)-\epsilon$ for all $X\subseteq S$. 
 This will be constructed from a family of $t\to\nu_E$ for finite sets $E$. Let $n\in \mathbb{N}$ and let $E\subseteq S$ be a finite set such that  $\mu(E)\geq \mu(S)-1/n$.

The key idea is to define the formula $\phi_e^k=\wedge_{e\models\phi_j,\ j\leq k}\phi_j$ for every $e\in E$, and where $(\phi_j)_{j\in \mathbb{N}}$  is an enumeration of the formulas of $\cL$.  Then  for every finite set $X\subseteq S$,
we set
$\phi_X^k=\vee_{e\in X}\phi_e^k$, and we let $p_{X}^k:=\mu(\sem {\vee_{e\in X}\phi_e^k})$. \jo{Since $e\models \phi_e^k$, we have $\sem{\phi_e^k}\supseteq \prec_\epsilon^\mathcal{L}(\{e\})$,  
}
$$p_X^k=\mu(\sem { \vee_{e\in X}\phi_e^k})\geq  \mu(\prec_\epsilon^\mathcal{L}(X))\geq  \mu(X),$$ 
and $s\models\angs a\lbrace(\vee_{e\in X}\phi_e^k,{p_X^k)}\rbrace_{X\subseteq E}$ for all $k\geq 1$.
By hypothesis,  we get that  $t\models_\epsilon\angs a\lbrace(\vee_{e\in X}\phi_e^k,{p_X^k)}\rbrace_{X\subseteq E}$
 for all $k\geq 1$. Let $\nu_E^k$ be the associated transition.
Then for all $X\subseteq E$ and for all $k\in\mathbb{N}$:
$$\nu_E^k(\sem{\vee_{e\in X}\phi_e^k})\geq p_X^k -\epsilon\geq \mu(X)-\epsilon.$$

Now, $\sem{\vee_{e\in X}\phi_e^k}_\epsilon$ is decreasing to $\prec_\epsilon^\mathcal{L}(X)$ as $k$ goes to infinity. Since the systems we consider are finitely branching, we can define a transition $t\stackrel{a}{\rightarrow}\nu_E$ such that $\nu_E$ is the limit of a subsequence of $\lbrace\nu_E^k\rbrace_{k\in\mathbb{N}}$. That is, there exists an increasing function $\psi: \mathbb{N}\rightarrow\mathbb{N}$ such that for all set $Y\subseteq S$ we have $\mathrm{lim}_{k\rightarrow\infty}\nu_E^{\psi(k)}(Y)=\nu_E(Y)$. 
This implies that: $\nu_E(\prec_\epsilon^\mathcal{L}(X))\geq\mu(X)-\epsilon.$
We have proven the following: for all $s\stackrel{a}{\rightarrow}\mu$, for all $E\subseteq S$ finite, there exists $t\stackrel{a}{\rightarrow}\nu_E$ such that for any $X\subseteq E$ we have $\nu_E(\prec_\epsilon^\mathcal{L}(X))\geq\mu(X)-\epsilon$. 
Let $E_k,k\in\mathbb{N}$ be a growing sequence of finite subsets of $S$ such that $S=\cup_{k\in\mathbb{N}}E_k$. Again, since the system is finitely branching, let $\nu$ be the limit of a subsequence of $\lbrace\nu_{E_k}\rbrace_{k\in\mathbb{N}}$. As before, we get that for any $X\subseteq S$ finite, $\nu(\prec_\epsilon^\mathcal{L}(X))\geq\mu(X)-\epsilon$. 

\noindent
$[\sim_{\epsilon}=\sim^{\mathcal{L}^N_\neg}_\epsilon] $
It can be proven that $\sim_\epsilon^{\mathcal{L}^N_\neg}$ is an $\epsilon$-bisimulation by following the proof above and using the fact that $\sim_\epsilon^{\mathcal{L}^N_\neg}$ is a symmetric relation (which comes from the presence of negation in $\mathcal{L}^N_\neg$).
\end{proof}

\section{A Bisimulation Pseudo-Metric between PAs}

\subsection{The pseudo-metric $d$}\label{d-metric}
The notion of $\epsilon$-bisimulation induces a pseudo-metric on states of a PA, given by the smallest $\epsilon$ such that the states are $\epsilon$-bisimilar.
\begin{definition}[Bisimulation metric]
Given $s,t\in S$, let $d(s,t)=\inf\lbrace \epsilon\ |\ s\sim_\epsilon t\rbrace$.
\end{definition}
Using the finite branching of our systems, we can prove $d$ is a bisimulation pseudo-metric, i.e., states at distance zero are bisimilar. We now discuss the  computation of this distance between all states of a given PA. We  propose three  approaches, the first being exact, and the others approximate.  The two first compute the distance iteratively, updating a function $d_i:S\times S\rightarrow[0;1]$, in the same way as it was done for deterministic PAs in \cite{DesLavTra08}. This approach is close to the classical iterative algorithms for computing simulation and bisimulation on probabilistic systems, see \cite{Baier96b,Segala02}. It makes use of a network flow computation.
The algorithm for the first approach is the left one  in Fig.~\ref{algos}. 
 Given $d:S\times S\rightarrow[0;1]$, and $\epsilon\geq0$, let $R^d_\epsilon$ be the relation on $S\times S$ defined as: $R^d_\epsilon(s,t)$ iff $d(s,t)\leq\epsilon$.

\begin{figure}[htbp]
\begin{center}
\begin{minipage}[t]{68mm}
\paragraph{Algorithm $\mathcal{A}$}: exact computation\\
\textbf{Input:} A finite PA $\mathcal{S}=(S,Act,\mathcal{D})$.\\
\textbf{Output:} $d:S\times S\rightarrow[0;1]$.\\
\textbf{Method:}\\
Let $d_0(s,t)=0$ $\forall s,t\in S\times S$. Let $j=0$\\
Until $d_j=d_{j+1}$ do begin:\\
\phantom{m} $d_{j+1}=d_j$.\\
\phantom{m} For all $(s,t)\in S\times S$ do begin:\\
\phantom{mm} For all $a\in Act$ do begin:\\
\phantom{mmm} For all $s\stackrel{a}{\rightarrow}\mu$ do begin:\\
\phantom{mmmm} 
\begin{minipage}{49mm}
Let $d_{j+1}(s,t)$ be the smallest $\epsilon\in[0;1]$ s.t. $\exists \ t\stackrel{a}{\rightarrow}\nu$ such that the maximum flow of network $\mathcal{N}(\mu,\nu,R_\epsilon^{d_j})$ is~$\geq\mu(S)-\epsilon$.
\end{minipage}\\
\phantom{mm} end end end\\
\phantom{m}$j=j+1$
end
return $d_{j-1}$.
\end{minipage}\phantom{mm} 
\begin{minipage}[t]{70mm}
\paragraph{Algorithm $\mathcal{B}$}: computation  up to $1/n$ \\
\textbf{Input:} A finite PA $\mathcal{S}=(S,Act,\mathcal{D})$, $n\in\mathbb N$.\\
\textbf{Output:} $d:S\times S\rightarrow[0;1]$.\\
\textbf{Method:}\\
Let $d_0(s,t)=0$ $\forall s,t\in S\times S$. Let $j=0$\\
For $m=n-1$ to $0$ do begin:\\
Until $d_j=d_{j+1}$ do begin:\\
\phantom{m} $d_{j+1}=d_j$.\\
\phantom{m} For all $(s,t)$ s.t.\  $d(s,t)=0$ do begin:\\
\phantom{mm} For all $a\in Act$ do begin:\\
\phantom{mmm} For all $s\stackrel{a}{\rightarrow}\mu$ do begin:\\
\phantom{mmmm} 
\begin{minipage}{49mm}
If $\exists \ t\stackrel{a}{\rightarrow}\nu$ s.t.\  the maximum flow of network $\mathcal{N}(\mu,\nu,R_0^{d_j})$ is~$<\mu(S)-m\epsilon/n$ then let $d_{j+1}(s,t)=(m+1)\epsilon/n$.
\end{minipage}\\
\phantom{mm} end end end
\\
\phantom{m}$j=j+1$
end  end return $d_{j-1}$.
\end{minipage}
\end{center}
\caption{Computations of metric $d$ on $\mathcal{S}$}
\label{algos}
\end{figure}


\begin{proposition}\label{p-algo correct}
 Algorithm  $\mathcal{A}$ correctly outputs the distance $d$ between all pairs of states in $S$. Moreover, algorithm $\mathcal{A}$ runs in time $O(|S|^9\cdot|Act|\cdot l^2)$, where $l$ is the maximal number of transitions with the same label issued from a single state.
\end{proposition}

This algorithm is quite expensive, and hence we propose  other  approaches that approximate the distance.  The first one is a variation of algorithm $\mathcal A$ and is the right-hand algorithm of  Fig.~\ref{algos}. Let $1/n$ be the accuracy we are interested in, for $n\in\mathbb N$. \jo{The idea is to compute $m/n$-bisimulation  iteratively, for $m$ decreasing from $n-1$ to $0$.  These relations are decreasing as $\epsilon$-bisimilarity implies $\epsilon'$-bisimilarity for any $\epsilon'>\epsilon$. At each iteration of that loop, states whose distance have not been established yet have $d$-value 0.  At step $m<n$, these states will be given distance   $(m+1)/n$ if they are not $m/n$-bisimilar. } The relation consisting of states at zero distance will decrease at every $j$ step.  For every pair of states, the worst number of flow networks to be computed will be $n$ (this happens if the states are bisimilar).
Hence the algorithm runs in $O(|S|^5\cdot n\cdot|Act|\cdot l^2)$.  Of course, some values of $m$ can be ignored and we can save some time.

The last algorithm that we propose uses recent work of Zhang et al.~\cite{zhang2007flow}, to update efficiently the flow computation in  algorithm $\mathcal A$. 
The algorithm of Zhang et al.\  computes strong bisimularity on a probabilistic automaton in time $O(|S|\cdot m^2)$, where $m$ is the total number of transitions. By a slight generalization of this algorithm to our context, we can compute $\epsilon$-bisimularity on $\mathcal{S}$ in time $O(|S|\cdot m^2)$, for any given $\epsilon$. Using a dichotomic approach, given two states $s$ and $t$, we can compute $d(s,t)$ up to an additive approximation factor $\delta$ in time  $O(|S|\cdot m^2\cdot log(\delta))$, and thus we can compute $d$  up to an additive approximation factor $\delta$ in time  $O(|S|^3\cdot m^2\cdot log(\delta))$.

\subsection{The decayed distance $d^\lambda$}\label{s-decay dist}
In the previous definitions, differences in the far future have as much importance as those in the near future. We can relax the impact of the future, or~instead the impact of short term transitions, by using a decayed relaxation. Instead of a fixed relaxation of parameter $\epsilon$, we can ask for a relaxation that changes as we get further from the starting state. As we get deeper through the transitions, the parameter could get bigger, hence diminishing the importance of further differences, or symmetrically, we could make the parameter smaller.  In order to leave this flexible, we will use a function $\lambda: [0;1]\to [0;1]$.  If $x\leq \lambda(x)$, the future will be neglected whereas $x\geq \lambda(x)$ will  make the future more precise. If $\lambda$ is the identity, we get the previous notions.  We will  describe below a natural choice for $\lambda$, but first, let us define the new semantics.  We write  $\lambda^n$ for to the $n$-th self composition of function $\lambda$, $n\in\mathbb{N}$, and $\lambda^0$ is the identity function on $[0;1]$.
Also, $\mathcal{L}^N_n$ will be the set of formulas of $\mathcal{L}^N$ of depth at most  $n\in\mathbb{N}$.


\begin{definition}[The $(\epsilon,\lambda)$-semantics]   Let $\epsilon\in[-1;1]$, and $\lambda: [0;1]\to [0;1]$
The syntax is the one of $\mathcal{L}^N$. The semantics $\models_\epsilon^\lambda$ is defined similarly as for  $\mathcal{L}^N$ except for the modal operator. Given $\phi\in\mathcal{L}^N$, let $\sem \phi _\epsilon^\lambda=\lbrace s\in S|s\models_\epsilon^\lambda\phi\rbrace$. Given $s\in S$, $s\models_\epsilon^\lambda \angs a\lbrace(\phi_i,p_i)\rbrace_{i\in I}$ iff there exists a transition $s\stackrel{a}{\rightarrow}\mu$  such that for all $i\in I$, we have $\mu(\sem {\phi_i}_{\lambda(\epsilon)}^\lambda)\geq p_i-\epsilon$. This semantics induces the  relations $\prec^{\mathcal{L}^N_n,\lambda}_\epsilon$ and $\sim^{\mathcal{L}^{N,\neg}_n,\lambda}_\epsilon$ as in Def.~\ref{d:logicalSim}.
\end{definition}

The relations  $\prec^{\mathcal{L}^N_n,\lambda}_\epsilon$ and $\sim^{\mathcal{L}^{N,\neg}_n,\lambda}_\epsilon$ are defined for formulas of a given maximal depth $n$, because in order to compute $\models^\lambda_\epsilon$ on $S\times S$ for a given $\epsilon\in[0;1]$, we may have to compute the $\models^\lambda_{\lambda^n(\epsilon)}$ for all $n\in\mathbb{N}$. 

Most of the time, one wants to give less importance to the future. In these situations, the decay is called a discount and could be exponential, as in~\cite{Desharnais02,ferns2004metrics}. In our case, this would correspond to asking that  there is a constant $0< c<1$ such that $1-\lambda(\epsilon)=c(1-\epsilon)$, i.e. $\lambda(\epsilon)=1-c\cdot(1-\epsilon)$. 

The associated simulation and bisimulation  are variations of Def.~\ref{d:eSim}:


\begin{definition}[Order $n$ $(\epsilon,\lambda)$-bi/simulation]\label{d:DecayeSim}
Given $n\in\mathbb{N}$, an \emph{order} $n$ $(\epsilon,\lambda)$ \emph{simulation} on $\mathcal{S}$ is a decreasing sequence of relations $R_0,...,R_n$ on $S$ such that $R_0=S\times S$, and for all $i\in[1;n]$, whenever $sR_i t$, if $s\stackrel{a}{\rightarrow}\mu$, then there exists  $t\stackrel{a}{\rightarrow}\nu$ such that $\mu\,\mathcal{L}^{\lambda^{n-i}(\epsilon)}(R_{i-1})\nu$.  
We write $s\prec^{\lambda,n}_\epsilon t$ if there exists an order $n$ $(\epsilon,\lambda)$-simulation $R_n$ on $\mathcal{S}$ such that $s R_n t$, and we write $s\sim^{\lambda,n}_\epsilon t$ if $s\prec^{\lambda,n}_\epsilon t$ and $t\prec^{\lambda,n}_\epsilon s$.
\end{definition}


\begin{proposition}\label{p-equiv dec}
Let $n\in\mathbb{N}$. Then
$\prec^{\mathcal{L}^{N}_n,\lambda}_\epsilon \;=\;\prec^{\lambda,n}_\epsilon $ and 
$\sim^{L^{\mathcal{N},\neg}_n,\lambda}_\epsilon \  \;=\;\ \sim^{\lambda,n}_\epsilon $
\end{proposition}


Given $s,t\in S$, we define $d^\lambda(s,t)=\min\lbrace \lambda^{-n}(1)\ |\ s\sim^{\lambda,n}_{\lambda^{-n}(1)} t\rbrace$. We can compute the distance $d^\lambda$ using the algorithm of Fig.~\ref{algosDecay}.
\begin{figure}[htbp]
\begin{center}
\begin{minipage}{11cm}
\textbf{Algorithm $\mathcal{C}$}: Computation of the discounted metric $d_\lambda$ on $\mathcal{S}$\\
\textbf{Input:} A finite PA $\mathcal{S}=(S,Act,\mathcal{D})$ $N\in\mathbb{N}$.\\
\textbf{Output:} $d:S\times S\rightarrow[0;1]$.\\
\textbf{Method:}\\
Let $d_0(s,t)=1$ for all $s,t\in S\times S$. Let $\delta=1/N$.\\
Let $R_0=S\times S$.\\
For $n=1$ to $N$ do begin:\\
\phantom{mi} For all $(s,t)\in S\times S$ do begin:\\
\phantom{mmmm} For all $a\in Act$ do begin:\\
\phantom{mmmmmm} For all $s\stackrel{a}{\rightarrow}\mu$ do begin:\\
\phantom{mmmmmmmm} 
\begin{minipage}{90mm}
Let $R_{n+1}=R_n$.\\
If there exists no transition $t\stackrel{a}{\rightarrow}\nu$ such that the maximum flow of the network $\mathcal{N}(\mu,\nu,R_n)$ is greater than or equal to $\mu(S)-(1-n\cdot\delta)$, then: let $d(s,t)=1-n\cdot\delta$,  and  let  $R_{n+1}=R_{n+1}-\lbrace (s,t)\rbrace$.
\end{minipage}\\
\phantom{mi} end end end
end\\
return $d$.
\end{minipage}
\end{center}
\caption{Computation of metric $d_\lambda$ on $\mathcal{S}$}
\label{algosDecay}
\end{figure}

\begin{proposition}
Algorithm $\mathcal{C}$ of Fig.~\ref{algosDecay} runs in time $O(|S|^5\cdot|Act|\cdot l^2)$.
\end{proposition}
\begin{proof}
Direct, using $O(|S|^3)$ flow network computations.
\end{proof}

\subsection{Comparison to other metrics on probabilistic systems}
During the past ten years, several metrics have been defined in the context of Probabilistic Automata or closely related models such as Labeled Markov Chains \cite{Desharnais99b,Ferns05,DesLavZhi06}, reactive probabilistic transition systems \cite{breugel05}, Markov Decision Processes \cite{ferns2004metrics,DesLavZhi06,Puterman:Book}, or more general game processes \cite{Alfaro07}. Most of these metrics are variations of the metric of \cite{Desharnais99b}. 
In \cite[Th.~4.6]{breugel05}, an equivalent metric is defined as a terminal coalgebra, using category theory. In \cite{Worrell01}, the same authors give an algorithm to compute in polynomial time this metric, relying on linear programming computation for a transshipment optimization problem. This approach is for deterministic models and it is applied in \cite{ferns2004metrics} and related papers to compute metrics between Markov Decision Processes. \jo{Most of} these algorithms introduce a decay factor to make the computation tractable. In \cite{Alfaro07} the authors consider metrics between systems which allow non determinism, but the complexity of the algorithms presented in \cite{chatterjee2008algorithms} to compute the metrics is at best PSPACE.

 The main difference between our metric and  those is  that differences along paths are not accumulated  in ours, even in the discounted metric: other metrics all involve comparing (among others) the probability of paths, and this makes these metrics straightforwardly different from ours, as we never multiply probability values. In \cite{Desharnais99b,ferns2004metrics}, the metric can be computed using a familly of functional expressions $\mathcal{F}^c$ from  states to $[0;1]$ that play the same role as the quantitative formulas of~\cite{Alfaro07}.
Given $s,t$ states of the system, the distance $d^c(s,t)$ is then defined as 
$d^c(s,t)=\sup_{f\in\mathcal{F}^c}|f^c(s)-f^c(t)|$. This distance is incomparable with ours, as shows the following example.

\begin{example}\label{e:differentmetrics}
 Let $0<\epsilon<1$, and consider the systems of Fig.~\ref{f:differentmetric}.
\begin{figure}
$$\xymatrix{
		&	s\ar[dl]_{a,1-\epsilon}\ar[d]^{a,\epsilon}&&t\ar[dl]_{a,1-\epsilon}\ar[d]^{a,\epsilon}&&&u\ar[d]_{a,1}\\
s_1	&	s_2\ar[d]^{b,1-\epsilon}							&t_1	&	t_2\ar[d]^{b,1 }	&&& u_1\ar[d]^{b,1/2}\\
		&	s_3						&								&t_3 								&&&u_3\\
}
\xymatrix{
&v\ar[dl]_{a,\epsilon}\ar[d]^{a,1-\epsilon}\\v_1&v_2\ar[d]^{b,1/2-\epsilon}\\&v_3
}$$
\caption{Our metric differs from those based on paths.}
\label{f:differentmetric}
\end{figure}
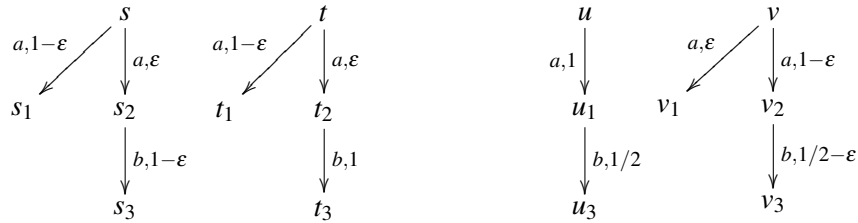
The distance between $s$ and $t$ is always greater with our distance, as $d^c(s,t)= c^2\cdot\epsilon^2\leq\epsilon=d(s,t)$. On the contrary, for $\epsilon<1/2$, there is some $c$ for which the distance between $u$ and $v$ is smaller with our distance.
Indeed, $d^c(u,v)>=c^2(\epsilon/2-\epsilon^2)$, whereas $d(u,v)=\epsilon$.  Hence we obtain $d(u,v)\leq d^c(u,v)$ by taking $c\in[0;1]$ such that $c^2(3\epsilon/2-\epsilon^2)>\epsilon$;  for example, with $\epsilon=1/4$ and $c^2>4/5$.  The example can be adapted for $\epsilon\geq 1/2$.
\end{example}




\subsection{The Metric $d$ on a Process Algebra}
We consider a process algebra on an extension of the model of PAs: the extension is that we distinguish between \emph{Input} and \emph{Output} action labels. As in \cite{Ferns05}, we assume a set of underlying labels $\Sigma$, and suppose that the labels of the PAs belong to a set $L=L!\cup L?$, where $L?=\lbrace a?|a\in\Sigma\rbrace$ and $L!=\lbrace a!|a\in\Sigma\rbrace$ are the sets of Input and Ouput labels respectivelly. Two PAs $\mathcal{S}_1=(S_1,Act_1,\mathcal{D}_1,s^0_1)$ and $\mathcal{S}_2=(S_2,Act_2,\mathcal{D}_2,s^0_2)$ will synchronize on labels in $Act_1\cap Act_2$.

We only present the operators for non-deterministic choice and parallel composition, as the other operators can be taken as in \cite{Ferns05}. 

\textbf{Non Deterministic Choice.} Let $\mathcal{S}_1,...,\mathcal{S}_k$ be PAs with respective state spaces $S_1,...,S_k$ and initial states $s^0_1,...,s^0_k$. Let $\lbrace a_1,...,a_l\rbrace\subseteq L$, and for each $a_i$ let $\lbrace \mu_{i,j}|j\in[1;n_i]\rbrace$ be a finite family of distributions on $\lbrace s^0_1,...,s^0_k\rbrace$. We define: $\mathcal{S}'={\ang + \ }_{i=1}^l\lbrace s\stackrel{a_i}{\rightarrow}\mu_{i,j}, j\in[1;n_i]\rbrace\lbrace \mathcal{S}_1,...,\mathcal{S}_k\rbrace$,
a PA whose state space is $S'=\lbrace s\rbrace\uplus_{i\in[1;k]} S_i$, and  initial state $s$.  Transitions from $s$ are all the $s\stackrel{a_i}{\rightarrow}\mu_{i,j}$. There may be several transitions from $s$ with the same label $a_i$. $\mathcal{S}'$ accepts the input label $a$ or outputs the label $a$, depending on $a\in L?$ or $a\in L!$.

\textbf{Parallel Composition.} Given PAs $\mathcal{S}_i=(S_1,Act_i,\mathcal{D}_i,s^0_i)$, $i=1,2$, we define the parallel composition $\mathcal{S}'=\mathcal{S}_1||\mathcal{S}_2$. The synchronisation is on labels in $Act_1\cap Act_2$. The state space of $\mathcal{S}'$ is $S'=S_1\times S_2$, with initial state  $s'^0=(s^0_1,s^0_2)$. The set of labels of $\mathcal{S}'$ is $Act'=Act_1\cup Act_2$. 
Given $\mu$ and $\nu$ two distributions on disjoint sets $S_1$ and $S_2$, given $X\subseteq S_1\times S_2$, let $\mu\otimes\nu(X)=\sum_{(s,t)\in X}\mu(s)\cdot\nu(t)$. Given states $s\in S_1$ and $t\in S_2$, we expect the following synchronized transitions:
\begin{itemize}
\item \textit{Synchronization on input labels in $Act_1\cap Act_2$}: $\forall a\in Act_1\cap Act_2$, if  $s\stackrel{a?}{\rightarrow}\mu$ and $t\stackrel{a?}{\rightarrow}\nu$, then there is a transition $(s,t)\stackrel{a?}{\rightarrow}\mu\otimes\nu$ on $\mathcal{S}$.
 \item \textit{Synchronization in the Output/Input}: if $s\stackrel{a!}{\rightarrow}\mu$ and $ t\stackrel{a?}{\rightarrow}\nu$, then there is a transition $(s,t)\stackrel{a!}{\rightarrow}\mu\otimes\nu$ (and symmetrically)
\item \textit{Asynchronous evolution on labels in $Act_1\setminus Act_2$}: given $a\in Act_1\setminus Act_2$ 
if $s\stackrel{a}{\rightarrow}\mu$, then  $(s,t)\stackrel{a}{\rightarrow}\mu\otimes\delta_t$, where $\delta_t$ is the Dirac distribution on  $t$ (symmetrically on $S_2$). 
\end{itemize}

We prove that the distance $d$ is non expansive with respect to the parallel operator: when composing two processes with a third one, the distance does not increase. Non expansiveness with respect to other operators is more common.
\begin{theorem}\label{t-process algebra}
$d(\mathcal{S}_1||\mathcal{S},\mathcal{S}_2||\mathcal{S})\leq d(\mathcal{S}_1,\mathcal{S}_2)$.
\end{theorem}
\begin{proof}
 Let $s,s_1,s_2$ be three states of $\mathcal{S},\mathcal{S}_1$ and $\mathcal{S}_2$ respectivelly. Let $\epsilon\geq0$, and suppose $s_1\sim_\epsilon s_2$. We want to prove that $s_1\times s$, which is a state of $\mathcal{S}_1||\mathcal{S}$, is $\epsilon$-bisimilar to $s_2\times s$, a state of $\mathcal{S}_2||\mathcal{S}$. We prove by induction on $n\in\mathbb{N}$ that if $s_1\sim^n_\epsilon s_2$, then $s_1\times s\sim^n_\epsilon s_2\times s$. We will use the following notations: Given $\phi\in\mathcal{L}^N$, and $j\in\lbrace 1,2\rbrace$:
 
\[\sem{\phi}_{\mathcal{S}_j||\mathcal{S}}=\lbrace s_j\times s|s_j\in S_j,\ s\in S,\ \mbox{and}\ s_j\times s\models\phi\]
Where the semantics is taken on the PA $\mathcal{S}_j||\mathcal{S}$. Let
\[\LEFT_j({\sem{\phi}}_{\mathcal{S}_j||\mathcal{S}})=\lbrace v\in S_j|\exists u\in S\ s.t.\ u\times v\models\phi\rbrace\]
Given $v\in S$, let
\[\RIGHT_j(v,{\sem{\phi}}_{\mathcal{S}_j||\mathcal{S}})=\lbrace u\in S_1| u\times v\models\phi\rbrace\]

The key case is when $\phi=\angs a\lbrace(\phi_i,p_i)\rbrace_{i\in I}\in\mathcal{L}^{N}$, with depth $n$. Suppose $s_1\times s\models\phi$. Then there exists a transition $s_1\times s\stackrel{a}{\rightarrow}\mu_1\otimes\nu$ on $\mathcal{S}'$ such that for all $i\in I$, $(\mu_1\otimes\nu)(\sem {\phi_i})\geq p_i$.

By hypothesis, $s_1\sim^n_\epsilon s_2$. Hence, there exists a transition $s_2\stackrel{a}{\rightarrow}\mu_2$ such that $\mu_1\mathcal{L}^\epsilon(\sim_\epsilon)\mu_2$. 

We know that for all $i\in I$, 
\[\mu_1\otimes\nu(\sem{\phi_i})=\sum_{v\in \LEFT_1(\sem{\phi_i})}\nu(v)\cdot\sum_{u\in \RIGHT_1(v,\sem{\phi_i})}\mu_1(v).\]
Given $v\in \LEFT(\sem{\phi_i})$ (hence $v\in S$), we know that $\mu_2(\sim^n_\epsilon \RIGHT_1(v,\sem{\phi_i}))\geq \mu_1(\RIGHT_1(v,\sem{\phi_i}))-\epsilon$. 

Moreover, by induction hypothesis, $\sim^n_\epsilon \RIGHT_1(v,\sem{\phi_i})\subseteq \RIGHT_2(v,{\sem{\phi_i}}_\epsilon$. Indeed, 
\[\sim^n_\epsilon \RIGHT_1(v,\sem{\phi_i})\cap S_2=\sim^n_\epsilon\lbrace u\in S_2|\exists u'\in S_1\ s.t.\ u\sim^n_\epsilon u'\ \mbox{and}\ u'\models\phi_i.\]
We get the result since by induction hypothesis, if $u\sim^n_\epsilon u'$ and $s'\in S$, we have $u\times s'\sim^n_\epsilon u'\times s'$.

This implies that $\mu_2(\RIGHT_2(v,{\sem{\phi_i}}_\epsilon))\geq \mu_1(\RIGHT_1(v,\sem{\phi_i}))-\epsilon$. Finally, 
$$\sum_{v\in \LEFT_2(\sem{\phi_i}_\epsilon)}\nu(v)\cdot\sum_{u\in \RIGHT_2(v,\sem{\phi_i}_\epsilon)}\mu_2(v)\geq \sum_{v\in \LEFT_1(\sem{\phi_i})}\nu(v)\cdot\sum_{u\in \RIGHT_1(v,\sem{\phi_i})}\mu_1(v) -\epsilon.$$ Hence $\mu_2\otimes\nu(\sem{\phi_i}_\epsilon)\geq \mu_1\otimes\nu(\sem{\phi_i})-\epsilon$. This proves that $s_2\times s\models_\epsilon\phi$.
\end{proof}

\section{Examples}

We build a benchmark set of deterministic PAs to compare distances. The processes are variations of a basic one from which we delete some transitions. The state space of the basic PA is a square grid of $n\times n$. The set of actions is $\lbrace a\rbrace$.	All the computations on the state indices  are done modulo $n$:  the grid is a torus.
 The $a$-transitions from state $(i,j)$ are as~follows:
\begin{center}
\begin{tabular}{|c||c|c|c|c|}\hline
~from $(i,j)$ to~~ 	& $\;(i,j-1):0.1\;$ & $\;(i,j+1):0.5\;$ & $\;(i-1,j):0.25\;$& $\;(i+1,j):0.15\;$\\\hline
\end{tabular}
\end{center}
This basic  PA is compared to variations of it obtained by deleting in some~states the transition of label $a$ (to all successors).  Note that the basic process is bisimilar to the one-state process that can do $a$ with probability 1. 
We consider the distances between states with same indices of the different systems.  Fig.~\ref{fig:ex} 
\begin{figure}[tb]
    \centering
    \begin{tabular}{ccccccc}
   
      \includegraphics[height=7cm]{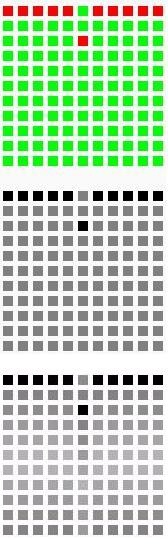}& $  $\includegraphics[height=7cm]{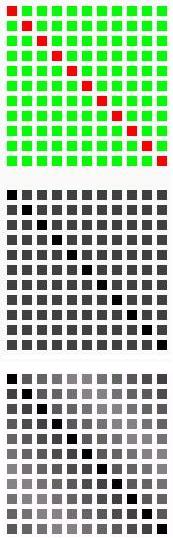}& $  $\includegraphics[height=7cm]{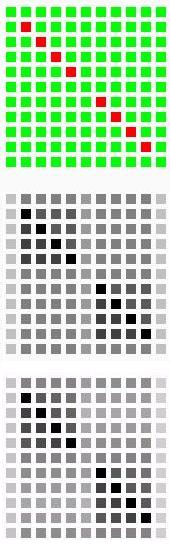}& 
      \includegraphics[height=7cm]{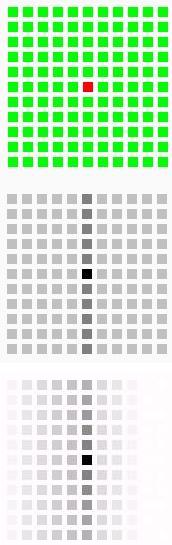}& $   $\includegraphics[height=7cm]{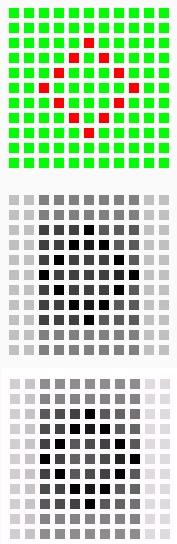}& 
      \includegraphics[height=7cm]{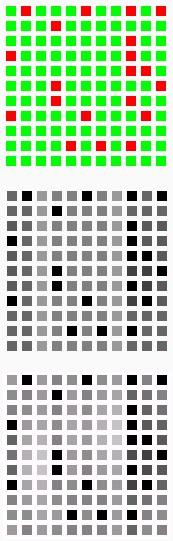}\\
    \end{tabular}
    \caption{The red entries of the top grid are states that have lost their $a$-transition. The two other grids have at entry $(i,j)$ the distance between state $(i,j)$ of top grid to state $(i,j)$ in the basic PA. The darker is the entry, the bigger the distance. The middle grid illustrates the basic distance, the other the decayed one.
    \label{fig:ex}}
\end{figure}
illustrates some PAs and the impact of the deletion of transitions  for the two distances that we defined. 
The distance $d^\lambda$ is illustrated in the bottom grids of the figure. The  linear function for $\lambda$ is the following, with $N=20$. We take   $\delta=1/N$, and let $\lambda(x)=x+\delta$ if $x\in[0;1-\delta]$, and $\lambda(x)=1$ if $x\in[1-\delta;1]$.  One can observe that the decay distances fade out when further from the difference, whereas $d$ is more constant. 

It would be nice to compare these grids  with  others obtained from other known metrics.  We leave that for future work, as we have no implementation of other metrics that can handle more than 25 states.

\section{Conclusion}
We presented relaxed notions of $\epsilon$-simulation and $\epsilon$-bisimulation. When $\epsilon=0$ we retrieve the usual notions of bisimulation and simulation on PAs. We gave logical characterisations of these notions and algorithms  to compute in PTIME two corresponding pseudo-metrics, one that discounts the future, and one that does not. We showed that our distance is not expansive with respect to process algebra operators. We also showed that  the basic logic $\cL^{(\neg)}$ characterises a  notion weaker than  $\epsilon$-(bi)simulation, called \emph{a priori $\epsilon$-(bi)simulation}.  Interestingly, we have proven this notion NP-difficult to decide. Further work includes relaxing what is called  probabilistic bisimulation and studying the associated distances; implementing our third proposal of algorithm to compute $d$, using a modification of the algorithm of \cite{zhang2007flow}; investigating further the weaknesses and strengths of the different metrics defined so far.


\bibliography{main}

\begin{thebibliography}{10}
\providecommand{\bibitemdeclare}[2]{}
\providecommand{\urlprefix}{Available at }
\providecommand{\url}[1]{\texttt{#1}}
\providecommand{\href}[2]{\texttt{#2}}
\providecommand{\urlalt}[2]{\href{#1}{#2}}
\providecommand{\doi}[1]{doi:\urlalt{http://dx.doi.org/#1}{#1}}
\providecommand{\bibinfo}[2]{#2}

\bibitemdeclare{inproceedings}{Alfaro07}
\bibitem{Alfaro07}
\bibinfo{author}{Luca de~Alfaro}, \bibinfo{author}{Rupak Majumdar},
  \bibinfo{author}{Vishwanath Raman} \& \bibinfo{author}{Marielle Stoelinga}
  (\bibinfo{year}{2007}): \emph{\bibinfo{title}{Game Relations and Metrics}}.
\newblock In: {\sl \bibinfo{booktitle}{LICS '07}}. \bibinfo{publisher}{IEEE
  Computer Society}, pp. \bibinfo{pages}{99--108}, \doi{10.1109/LICS.2007.22}.

\bibitemdeclare{inproceedings}{Baier96b}
\bibitem{Baier96b}
\bibinfo{author}{Christel Baier} (\bibinfo{year}{1996}):
  \emph{\bibinfo{title}{Polynomial Time Algorithms for Testing Probabilistic
  Bisimulation and Simulation}}.
\newblock In: {\sl \bibinfo{booktitle}{CAV'96}}. {\sl \bibinfo{series}{LNCS}}
  \bibinfo{volume}{1102}, pp. \bibinfo{pages}{38--49},
  \doi{10.1007/3-540-61474-5\_57}.

\bibitemdeclare{inproceedings}{breugelSW07}
\bibitem{breugelSW07}
\bibinfo{author}{Franck van Breugel}, \bibinfo{author}{Babita Sharma} \&
  \bibinfo{author}{James Worrell} (\bibinfo{year}{2007}):
  \emph{\bibinfo{title}{Approximating a Behavioural Pseudometric Without
  Discount for Probabilistic Systems.}}
\newblock In: {\sl \bibinfo{booktitle}{FOSSACS}}. {\sl \bibinfo{series}{LNCS}}
  \bibinfo{volume}{4423}, \bibinfo{publisher}{Springer}, pp.
  \bibinfo{pages}{123--137}, \doi{10.1007/978-3-540-71389-0\_10}.

\bibitemdeclare{inproceedings}{Worrell01}
\bibitem{Worrell01}
\bibinfo{author}{Franck van Breugel} \& \bibinfo{author}{James Worrell}
  (\bibinfo{year}{2001}): \emph{\bibinfo{title}{An Algorithm for Quantitative
  Verification of Probabilistic Transition Systems}}.
\newblock In: {\sl \bibinfo{booktitle}{CONCUR '01}}.
  \bibinfo{publisher}{Springer-Verlag}, pp. \bibinfo{pages}{336--350},
  \doi{10.1007/3-540-44685-0\_23}.

\bibitemdeclare{article}{breugel05}
\bibitem{breugel05}
\bibinfo{author}{Franck van Breugel} \& \bibinfo{author}{James Worrell}
  (\bibinfo{year}{2004}): \emph{\bibinfo{title}{A behavioural pseudometric for
  probabilistic transition systems}}.
\newblock {\sl \bibinfo{journal}{Theor. Comput. Sci.}}
  \bibinfo{volume}{331}(\bibinfo{number}{1}), pp. \bibinfo{pages}{115--142},
  \doi{10.1007/3-540-48224-5\_35}.

\bibitemdeclare{inproceedings}{Segala02}
\bibitem{Segala02}
\bibinfo{author}{Stefano Cattani} \& \bibinfo{author}{Roberto Segala}
  (\bibinfo{year}{2002}): \emph{\bibinfo{title}{Decision Algorithms for
  Probabilistic Bisimulation}}.
\newblock In: {\sl \bibinfo{booktitle}{CONCUR '02}}.
  \bibinfo{publisher}{Springer-Verlag}, pp. \bibinfo{pages}{371--385},
  \doi{10.1007/3-540-45694-5\_25}.

\bibitemdeclare{inproceedings}{CSKN05}
\bibitem{CSKN05}
\bibinfo{author}{Stefano Cattani}, \bibinfo{author}{Roberto Segala},
  \bibinfo{author}{Marta Kwiatkowska} \& \bibinfo{author}{Gethin Norman}
  (\bibinfo{year}{2005}): \emph{\bibinfo{title}{Stochastic transition systems
  for continuous state spaces and non-determinism}}.
\newblock In: {\sl \bibinfo{booktitle}{FOSSACS'05}}. {\sl
  \bibinfo{series}{LNCS}} \bibinfo{volume}{3441}, \bibinfo{publisher}{Springer
  Verlag}, pp. \bibinfo{pages}{125--139}, \doi{10.1007/978-3-540-31982-5\_8}.

\bibitemdeclare{article}{chatterjee2008algorithms}
\bibitem{chatterjee2008algorithms}
\bibinfo{author}{K.~Chatterjee}, \bibinfo{author}{Luca de~Alfaro},
  \bibinfo{author}{Rupak Majumdar} \& \bibinfo{author}{Vishwanath Raman}
  (\bibinfo{year}{2010, to appear.}): \emph{\bibinfo{title}{{Algorithms for
  game metrics}}}.
\newblock {\sl \bibinfo{journal}{LMCS: Logical Methods in Computer Science}}
  \doi{10.2168/LMCS-6(3:13)2010}.

\bibitemdeclare{article}{DArg09}
\bibitem{DArg09}
\bibinfo{author}{Pedro~R. D'Argenio}, \bibinfo{author}{Nicol\'as Wolovick},
  \bibinfo{author}{Pedro~S\'anchez Terraf} \& \bibinfo{author}{Pablo Celayes}
  (\bibinfo{year}{2009}): \emph{\bibinfo{title}{Nondeterministic labeled
  {Markov} processes: bisimulations and logical characterization}}.
\newblock {\sl \bibinfo{journal}{QEST'09}} , pp.
  \bibinfo{pages}{11--20}\doi{10.1109/QEST.2009.17}.

\bibitemdeclare{article}{Desharnais02}
\bibitem{Desharnais02}
\bibinfo{author}{Jos\'ee Desharnais}, \bibinfo{author}{Abbas Edalat} \&
  \bibinfo{author}{Prakash Panangaden} (\bibinfo{year}{2002}):
  \emph{\bibinfo{title}{Bisimulation for Labeled {Markov} Processes}}.
\newblock {\sl \bibinfo{journal}{Information and Computation}}
  \bibinfo{volume}{179}(\bibinfo{number}{2}), pp. \bibinfo{pages}{163--193},
  \doi{10.1.1.16.5653}.

\bibitemdeclare{inproceedings}{Desharnais99b}
\bibitem{Desharnais99b}
\bibinfo{author}{Jos\'ee Desharnais}, \bibinfo{author}{Vineet Gupta},
  \bibinfo{author}{R.~Jagadeesan} \& \bibinfo{author}{P.~Panangaden}
  (\bibinfo{year}{1999}): \emph{\bibinfo{title}{Metrics for Labeled {Markov}
  Processes}}.
\newblock In: {\sl \bibinfo{booktitle}{CONCUR'99}}. \bibinfo{series}{LNCS},
  \bibinfo{publisher}{Springer-Verlag}, pp. \bibinfo{pages}{258--273},
  \doi{10.1007/3-540-48320-9\_19}.

\bibitemdeclare{conference}{DesLavTra08}
\bibitem{DesLavTra08}
\bibinfo{author}{Jos\'ee Desharnais}, \bibinfo{author}{Fran\c{c}ois Laviolette}
  \& \bibinfo{author}{Mathieu Tracol} (\bibinfo{year}{2008}):
  \emph{\bibinfo{title}{Approximate analysis of probabilistic processes: logic,
  simulation and games}}.
\newblock In: {\sl \bibinfo{booktitle}{QEST'08}}. \bibinfo{publisher}{IEEE
  Computer Society}, pp. \bibinfo{pages}{264--273}, \doi{10.1109/QEST.2008.42}.

\bibitemdeclare{inproceedings}{DesLavZhi06}
\bibitem{DesLavZhi06}
\bibinfo{author}{Jos\'ee Desharnais}, \bibinfo{author}{Fran\c{c}ois Laviolette}
  \& \bibinfo{author}{Sami Zhioua} (\bibinfo{year}{2006}):
  \emph{\bibinfo{title}{Testing Probabilistic Equivalence Through Reinforcement
  Learning}}.
\newblock In: {\sl \bibinfo{booktitle}{FSTTCS'06}}. {\sl
  \bibinfo{series}{LNCS}} \bibinfo{volume}{4337},
  \bibinfo{publisher}{Springer}, pp. \bibinfo{pages}{236--247},
  \doi{10.1007/11944836\_23}.

\bibitemdeclare{conference}{ferns2004metrics}
\bibitem{ferns2004metrics}
\bibinfo{author}{Norm Ferns}, \bibinfo{author}{Prakash Panangaden} \&
  \bibinfo{author}{Doina Precup} (\bibinfo{year}{2004}):
  \emph{\bibinfo{title}{{Metrics for finite Markov decision processes}}}.
\newblock In: {\sl \bibinfo{booktitle}{UAI'04}}. \bibinfo{organization}{AUAI
  Press Arlington, Virginia, United States}, pp. \bibinfo{pages}{162--169},
  \doi{10.1.1.87.9485}.

\bibitemdeclare{inproceedings}{Ferns05}
\bibitem{Ferns05}
\bibinfo{author}{Norman Ferns}, \bibinfo{author}{Prakash Panangaden} \&
  \bibinfo{author}{Doina Precup} (\bibinfo{year}{2005}):
  \emph{\bibinfo{title}{Metrics for {Markov} Decision Processes with Infinite
  State Spaces}}.
\newblock In: {\sl \bibinfo{booktitle}{UAI'05}}. \bibinfo{publisher}{AUAI
  Press}, p. \bibinfo{pages}{201}, \doi{10.1007/BF01908587}.

\bibitemdeclare{book}{garey1979computers}
\bibitem{garey1979computers}
\bibinfo{author}{Michael~R. Garey} \& \bibinfo{author}{David~S. Johnson}
  (\bibinfo{year}{1979}): \emph{\bibinfo{title}{{Computers and Intractability:
  A Guide to the Theory of NP-completeness}}}.
\newblock \bibinfo{publisher}{W. H. Freeman \& Co.}, \bibinfo{address}{New
  York, NY, USA}.

\bibitemdeclare{inproceedings}{Giacalone90}
\bibitem{Giacalone90}
\bibinfo{author}{Alessandro Giacalone}, \bibinfo{author}{Chi-Chang Jou} \&
  \bibinfo{author}{Scott~A. Smolka} (\bibinfo{year}{1990}):
  \emph{\bibinfo{title}{Algebraic Reasoning for Probabilistic Concurrent
  Systems}}.
\newblock In: {\sl \bibinfo{booktitle}{Proc. IFIP TC2 Working Conference on
  Programming Concepts and Methods}}. \bibinfo{publisher}{North-Holland}, pp.
  \bibinfo{pages}{443--458}, \doi{10.1.1.56.3664}.

\bibitemdeclare{book}{HermannsIMC02}
\bibitem{HermannsIMC02}
\bibinfo{author}{Holger Hermanns} (\bibinfo{year}{2002}):
  \emph{\bibinfo{title}{Interactive Markov chains, and the quest for quantified
  quality}}.
\newblock \bibinfo{publisher}{Springer-Verlag}, \bibinfo{address}{Berlin,
  Heidelberg}, \doi{10.1007/3-540-45804-2}.

\bibitemdeclare{article}{Larsen91}
\bibitem{Larsen91}
\bibinfo{author}{Kim~G. Larsen} \& \bibinfo{author}{Arne Skou}
  (\bibinfo{year}{1991}): \emph{\bibinfo{title}{Bisimulation through
  Probablistic Testing}}.
\newblock {\sl \bibinfo{journal}{Information and Computation}}
  \bibinfo{volume}{94}, pp. \bibinfo{pages}{1--28}, \doi{10.1.1.158.9316}.

\bibitemdeclare{inproceedings}{Segala07}
\bibitem{Segala07}
\bibinfo{author}{Augusto Parma} \& \bibinfo{author}{Roberto Segala}
  (\bibinfo{year}{2007}): \emph{\bibinfo{title}{Logical characterisation of
  Bisimulations for Discrete Probabilistic Systems}}.
\newblock In: {\sl \bibinfo{booktitle}{FOSSACS'07}}. {\sl
  \bibinfo{series}{LNCS}} \bibinfo{volume}{4423},
  \bibinfo{publisher}{Springer-Verlag}, pp. \bibinfo{pages}{287--301},
  \doi{10.1007/978-3-540-71389-0\_21}.

\bibitemdeclare{book}{Puterman:Book}
\bibitem{Puterman:Book}
\bibinfo{author}{Martin~L. Puterman} (\bibinfo{year}{1994}):
  \emph{\bibinfo{title}{Markov {D}ecision {P}rocesses: {D}iscrete Stochastic
  Dynamic Programming}}.
\newblock \bibinfo{publisher}{Wiley}.

\bibitemdeclare{inproceedings}{Segala94}
\bibitem{Segala94}
\bibinfo{author}{Roberto Segala} \& \bibinfo{author}{Nancy Lynch}
  (\bibinfo{year}{1994}): \emph{\bibinfo{title}{Probabilistic Simulations for
  Probabilistic Processes}}.
\newblock In: {\sl \bibinfo{booktitle}{CONCUR'94}}. {\sl
  \bibinfo{series}{LNCS}} \bibinfo{volume}{836},
  \bibinfo{publisher}{Springer-Verlag}, pp. \bibinfo{pages}{481--496},
  \doi{10.1007/BFb0015027}.

\bibitemdeclare{inproceedings}{SegalaT07}
\bibitem{SegalaT07}
\bibinfo{author}{Roberto Segala} \& \bibinfo{author}{Andrea Turrini}
  (\bibinfo{year}{2007}): \emph{\bibinfo{title}{Approximated Computationally
  Bounded Simulation Relations for Probabilistic Automata}}.
\newblock In: {\sl \bibinfo{booktitle}{CSF'07}}. pp. \bibinfo{pages}{140--156},
  \doi{10.1007/3-540-45694-5\_25}.

\bibitemdeclare{inproceedings}{zhang2007flow}
\bibitem{zhang2007flow}
\bibinfo{author}{Lijun Zhang}, \bibinfo{author}{Holger Hermanns},
  \bibinfo{author}{Friedrich. Eisenbrand} \& \bibinfo{author}{David~N. Jansen}
  (\bibinfo{year}{2007}): \emph{\bibinfo{title}{{Flow faster: efficient
  decision algorithms for probabilistic simulations}}}.
\newblock In: {\sl \bibinfo{booktitle}{TACAS'07}}. {\sl \bibinfo{series}{LNCS}}
  \bibinfo{volume}{4424}, \bibinfo{publisher}{Springer-Verlag}, pp.
  \bibinfo{pages}{155--169}, \doi{10.1007/978-3-540-71209-1\_14}.

\end{thebibliography}

\bibliographystyle{eptcs}


\end{document}